\documentclass[11pt]{article}
\usepackage[margin=1in]{geometry}

\usepackage{amssymb,amsmath,amsthm}
\usepackage{url}
\usepackage{mathpazo}
\usepackage{complexity}
\usepackage{enumerate}
\usepackage{nicefrac}

\usepackage{setspace}
\usepackage{hyperref}\hypersetup{colorlinks=true,linkcolor=black,citecolor=[rgb]{0.5,0,0},urlcolor=cyan}

\newtheorem{theorem}{Theorem}[section]
\newtheorem{cor}[theorem]{Corollary}
\newtheorem{lemma}[theorem]{Lemma}
\newtheorem{prop}[theorem]{Proposition}
\newtheorem{hypothesis}[theorem]{Hypothesis}
\newtheorem{claim}[theorem]{Claim}
\newtheorem{remark}[theorem]{Remark}

\newtheorem{definition}[theorem]{Definition}

\usepackage{xspace} 

\renewcommand{\epsilon}{\varepsilon}
\renewcommand{\tilde}{\widetilde}
\newcommand{\abs}[1]{\left|#1\right|}
\newcommand{\sett}[2]{\left\{ #1 \left| \; \vphantom{#1 #2} \right. #2  \right\}} 

\newcommand{\N}{\mathbb{N}}

\newcommand{\Exp}{\mathbb{E}}

\newcommand{\maxcover}{\mbox{\sf MaxCover}\xspace}
\newcommand{\clique}{$k$-$\mathsf{Clique}$\xspace}
\newcommand{\biclique}{$k$-$\mathsf{Biclique}$\xspace}
\newcommand{\oneside}{$k$-$\mathsf{One}$-$\mathsf{Sided}$-$\mathsf{Biclique}$\xspace}
\newcommand{\steiner}{$k$-$\mathsf{Steiner Orientation}$\xspace}

\newcommand{\cnf}{$\mathsf{CNF}$\xspace}
\newcommand{\TGC}{$\mathsf{TGC}$\xspace}
\newcommand{\DPF}{$\mathsf{DPCPF}$\xspace}
\newcommand{\setcover}{\textsf{SetCover}\xspace}

\newcommand{\Cov}{\mathsf{Col}}
\newcommand{\dist}{\Delta}
\newcommand{\neighb}{\mathcal{N}}
\newcommand{\sat}{\mathsf{SAT}}

\renewcommand\eth{$\mathsf{ETH}$}
\newcommand\seth{$\mathsf{SETH}$}
\newcommand\Wone{$\mathsf{W[1]}\neq\mathsf{FPT}$}

\title{{On Hardness of Approximation of\\ Parameterized Set Cover and Label Cover:\\  Threshold Graphs from Error Correcting Codes}}
\author{Karthik C.\ S.\footnote{This work was supported by  the  Israel Science Foundation (grant number 552/16) and the Len Blavatnik and the Blavatnik Family foundation. } \\
Tel Aviv University\\
\texttt{karthiks@mail.tau.ac.il}\\ \  \and Inbal Livni-Navon\footnote{This work was supported by Irit Dinur's ERC-CoG grant 772839.}\\ Weizmann Institute of Science\\
\texttt{inbal.livni@weizmann.ac.il}\\
}

\date{}

\begin{document}

\maketitle

\begin{abstract}
In the $(k,h)$-\setcover problem, we are given a collection $\mathcal{S}$ of sets over a universe $\mathcal{U}$, and the goal is to distinguish between the case that $\mathcal{S}$ contains $k$ sets which cover $\mathcal{U}$, from the case that at least $h$ sets in $\mathcal{S}$ are needed to cover $\mathcal{U}$. Lin (ICALP'19) recently showed a gap creating reduction from the $(k,k+1)$-\setcover problem on universe of size $O_k(\log |\mathcal{S}|)$ to the $\left(k,\sqrt[k]{\frac{\log |\mathcal{S}|}{\log\log |\mathcal{S}|}}\cdot k\right)$-\setcover problem on universe of size $|\mathcal{S}|$. 
In this paper, we prove a more scalable version of his result: given any error correcting code $C$ over alphabet $[q]$, rate $\rho$, and relative distance $\delta$,  we use $C$ to create a reduction from the $(k,k+1)$-\setcover problem on universe $\mathcal{U}$ to the $\left(k,\sqrt[2k]{\frac{2}{1-\delta}}\right)$-\setcover problem on universe of size $\frac{\log |\mathcal{S}|}{\rho}\cdot |\mathcal{U}|^{q^k}$.  \\

Lin established his result by composing the input \setcover instance (that has no gap) with a special threshold graph constructed from extremal combinatorial object called universal sets, resulting in a final \setcover instance with gap. Our reduction follows along the exact same lines, except that we generate the threshold graphs specified by Lin simply using the basic properties of the error correcting code $C$. We further show that one can recover the precise result of Lin by using a code which also achieves optimal parameters as a perfect hash function. \\

\begin{sloppypar}We use the same threshold graphs mentioned above to prove inapproximability results, under W[1]$\neq$FPT and ETH, for the parameterized label cover problem called $k$-\maxcover introduced by Chalermsook et al. (FOCS'17; SICOMP'20). Our inapproximaiblity results match the bounds obtained by Karthik et al. (STOC'18; JACM'19), although their proof framework is very different, and involves generalization of the `distributed PCP framework'. To the best of our knowledge, prior to this work, it was not clear how to adopt the proof strategy of Lin to prove inapproximability results for   $k$-\maxcover.\end{sloppypar}
\end{abstract}
\clearpage

\section{Introduction}\label{sec:intro}

Many optimization problems that we care about are \NP-Hard. Two typical ways to cope with \NP-Hardness are to design approximation algorithms or fixed parameter algorithms. For example, consider the classic \setcover problem which was shown in the seminal work of Karp \cite{K72} to be \NP-Hard. Researchers coped with this hardness by designing approximation algorithms \cite{C79,S95,S96} and studying its fixed parameter tractability. 

Nevertheless,  in many cases we have that  even finding a good approximate solution to \NP-Hard optimization problems is still \NP-Hard (and such results are proved using the  celebrated PCP theorem \cite{AS98,ALMSS98,D07}). In similar spirit, it is also possible to show that many natural parameterized variants of \NP-Hard problems are \textsf{W[1]}-Hard and thus not fixed parameter tractable. Case in point, it was shown that on one hand it is \NP-Hard to approximate \setcover below logarithmic factors \cite{F98,DS14} and on the other hand that the \setcover problem parameterized by the solution size is not fixed parameter tractable assuming \textsf{W[2]}$\neq$\textsf{FPT} \cite{DF95}. Thus, one may further try to cope with both hardness of approximation and fixed parameter intractability, simultaneously, by the design of fixed parameter approximation algorithms. In this paper, we are interested in the recently emerging theory of \emph{fixed parameter inapproximability}, i.e., the subarea formed by the intersection of hardness of approximation and parameterized complexity.

The results in fixed parameter inapproximability can be broadly divided into two parts. First, we have the results obtained under non-gap assumptions such as \Wone, \eth, and \seth\ \cite{IP01,IPZ01}. The main difficulty addressed in these results is generating a gap, i.e., we focus on how to start from a hard problem with no gap, say \clique, and reduce it to a problem of interest while generating a non-trivial gap in the process. We elaborate below on these set of results.  The other collection of results in fixed parameter inapproximability are under gap
assumptions such as the Gap Exponential Time Hypothesis \cite{MR16,D16} and Parameterized Inapproximability Hypothesis \cite{LRSZ20}. In these results the gap is inherent in the assumption,
and the challenge is to construct gap-preserving reductions.  These results are not the focus of this paper and we shall not elaborate further on them, and the interested reader may see the recent survey of Feldman et al. \cite{FKLM20} for more details.

There are two main techniques to generate the gap in the fixed parameter inapproximability literature\footnote{One additional technique that we do not address in this paper is due to  Wlodarczyk  \cite{W20}, who recently used a variant  of the gap amplification via graph products technique to prove hardness of approximation for connectivity problems,
including the \steiner problem.}. \vspace{0.05cm}

\noindent\textbf{Threshold Graph Composition.} The Threshold Graph Composition (\TGC) technique was introduced in the breakthrough work of Lin \cite{L18} to show the \textsf{W[1]}-Hardness of the \biclique problem via the inapproximability of the \oneside problem. This technique was later used to prove the first non-trivial inapproximability result for the $k$-\setcover problem \cite{CL19}, and is also the current technique used to prove the state-of-the-art inapproximability result for the same \cite{L19}. Moreover, the result on the \oneside problem in Lin \cite{L18} was used by Bhattacharyya et al. \cite{BBEGKLMM19} as the starting point to prove inapproximability results for problems in coding theory such as  the $k$-\textsf{Minimum Distance} problem (a.k.a.\ $k$-\textsf{Even Set} problem) and  the $k$-\textsf{Nearest Codeword}  problem, and also for lattice problems such as the $k$-\textsf{Shortest Vector} problem and the  $k$-\textsf{Nearest Vector} problem. 

At a very high level, in \TGC we compose an instance of the input problem that has no gap, with a \emph{threshold graph} (that is constructed oblivious to the input instance; see Section ~\ref{sec:proof} for the definition), to produce a gap instance of the desired problem. The main challenge here is to find the right way to compose the input and the threshold graph, although we remark that even the task of constructing the requisite threshold graphs is in many cases non-trivial. \vspace{0.05cm}

\noindent\textbf{Distributed PCP Framework.} The Distributed PCP Framework (\DPF) was introduced in the seminal work of Abboud et al.\ \cite{ARW17} and laid the foundation for a series of inapproximability results in the area of fine-grained complexity \cite{R18,C18,AR18,CGLRR19}. The framework was used by Karthik et al.\  \cite{KLM19} to rule out fixed parameter approximation algorithms for the $k$-\setcover problem. En route, they also provided inapproximability for $k$-\maxcover, a parameterized variant of the label cover problem which was introduced and identified by Chalermsook et al. \cite{CCKLMNT20} as a key intermediate gap problem to be studied in parameterized complexity. At a very high level, in \DPF, given an instance of the input problem with no gap, one first designs a protocol for a specific communication problem formulated based on the input problem, and then extracts an instance of the gap $k$-\maxcover problem from the transcript of the protocol. Finally, one designs a gap preserving reduction from the gap $k$-\maxcover problem to the   gap problem of interest. 

\noindent\textbf{Meeting Point of the Two Techniques.} While the aforementioned two techniques seems very different, rather surprisingly, they both yield very similar inapproximability results for the same problem: $k$-\setcover \cite{KLM19,L19}. This leads to the following natural question:

\begin{center}
\textit{Is there a unified technique to yield all inapproximability results\\ in parameterized complexity?}
\end{center}

More concretely, one can ask if it is possible to recover using the \DPF all the inapproximability results that are currently only obtained using the \TGC technique and vice-versa? In   \cite{KM19} the authors made the connection that if one could construct certain high dimensional extremal combinatorial objects then it is possible to prove the inapproximability of \oneside through \DPF (specifically by using the result of \cite{KLM19} on $k$-\maxcover). However, the construction of the desired combinatorial objects seem far from reach using current techniques.  In this paper, we look at the other direction of the question and address which results obtained in \DPF can now be obtained using the \TGC technique.

\subsection{Our Results}\label{sec:results}

Towards answering the raised question we use the \TGC technique to prove the following gap creating self reduction for the $k$-\maxcover problem. We start by describing the $k$-\maxcover problem (see Section~\ref{sec:problems} for a formal definition).

In the $k$-\maxcover problem we are given a bipartite graph $\Gamma=(V\cup W,E)$, where the vertex set is partitioned  as follows: $V=V_1\dot\cup\cdots \dot\cup V_k$ and $W=W_1\dot\cup\cdots \dot\cup W_t$. We denote by $|\Gamma|=|V|+|W|+|E|$. A labeling of $V$ is a $k$-tuple of vertices $(v_1,\ldots , v_k)\in V_1\times\cdots\times V_k$, and we say that  it \emph{covers} $W_j$ (for some $j\in[t]$) if $\exists w\in W_j$ which is a joint neighbor of all of $v_1,\ldots ,v_k$. We denote by $\maxcover(\Gamma)$ the maximal fraction of $W_j$ that can be simultaneously covered, i.e., $$\maxcover(\Gamma) :=  \max_{(v_1,\ldots , v_k)\in V_1\times\cdots\times V_k } \left(\Pr_{j \sim [t]}\left[ W_j \text{ is covered by } (v_1,\ldots , v_k)\right]\right).$$ 

It is easy to see that $k$-\maxcover is a parameterized variant of the classical label cover problem. Our first  result is a reduction, by only using an arbitrary error correcting code, from the exact $k$-\maxcover problem with a certain projection property, (which we call \emph{pseduo-projection} property, and is analogous to the standard projection property of label cover problem)  to  the gap  $k$-\maxcover problem.

\begin{theorem}[\maxcover Gap Creation using \TGC technique; Informal statement of Theorem~\ref{thm:gap}]\label{thm:gapintro}
Let $\Gamma_0$ be a $k$-\maxcover instance with the ``pseudo-projection'' property. Let $C$ be an error correcting code over alphabet set $[q]$ of block length $\ell$ and message length $\log_q |\Gamma_0|$.
Then there exists a reduction in time $O(\abs{\Gamma_0}\ell\cdot q^t)$ to a $k$-\maxcover instance $\Gamma$ of size $\abs{\Gamma_0}\ell\cdot q^t$. The new instance $\Gamma$ satisfies,
\begin{description}
\item[Completeness:] If $\maxcover(\Gamma_0)=1$, then $\maxcover(\Gamma)=1$,
\item[Soundness:] If $\maxcover(\Gamma_0)<1$, then $\maxcover(\Gamma)\le 1-\dist(C)$, 
\end{description}
where $\dist(C)$ is the relative distance of $C$.  
\end{theorem}

It is clear that in the above theorem, by taking \emph{any} code with constant relative distance bounded away from 0, we already obtain a $k$-\maxcover instance with constant gap. We elaborate more on the proof technique in the next subsection of the introduction, but for now discuss the context of the above result. 

We use the above gap creation theorem to show inapproximability results for $k$-\maxcover based on \eth\ and \Wone. In particular, we show in Theorem~\ref{thm:eth-hardness} (resp.\ Theorem~\ref{thm:wone}) that assuming \eth\ (resp.\ \Wone), there is no algorithm running in time $n^{o(k)}$ (resp.\ $F(k)\cdot n^{O(1)}$ time, for some computable function $F$), that can decide if a $k$-\maxcover instance has completeness 1 or soundness at most $\left(\frac{1}{n}\right)^{\nicefrac{1}{\poly(k)}}$, where $n=|V|$. 
The proofs of Theorems~\ref{thm:eth-hardness}~and~\ref{thm:wone} follows by first showing \eth-Hardness and \textsf{W[1]}-Hardness of exact $k$-\maxcover having the pseudo-projection property, and then applying Theorem~\ref{thm:gapintro}.

\noindent\textbf{Comparison to \cite{KLM19}.} We remark that the proof technique in \cite{KLM19} also gives us Theorem~\ref{thm:gapintro} with identical parameters (i.e., even using \DPF, we can create gap in $k$-\maxcover as in Theorem~\ref{thm:gapintro} using an arbitrary error correcting code). Therefore, our contribution is about proving the same result using \TGC technique. See 
Remark~\ref{rem:KLM} for more details. 

Our next contribution is a more scalable version of Lin's result on gap $k$-\setcover.

\begin{theorem}[Scalable version of \cite{L19}; Informal statement of Theorem~\ref{thm:setcover}]\label{thm:setcoverintro}
There is a polynomial time algorithm taking an instance  $(\mathcal{S},\mathcal{U})$ of $k$-\emph{\setcover} problem, and an error correcting code $C$  over alphabet set $[q]$ of block length $\ell$ and message length $\log_q |\mathcal{S}|$, and outputs an instance  $\mathcal{U}',\mathcal{S}'$ of $k$-\emph{\setcover} problem, of size $|\mathcal{S}'| = |\mathcal{S}|, |\mathcal{U}'|= \ell \cdot |\mathcal{U}|^{q^{k}}$ such that the following holds. 
\begin{description}
	\item[Completeness:] If  $(\mathcal{S},\mathcal{U})$ has a cover of size $k$, then so does  $(\mathcal{S}',\mathcal{U}')$,
	\item[Soundness:]   If $(\mathcal{S},\mathcal{U})$ does not have a cover of size $k$ then $(\mathcal{S}',\mathcal{U}')$ does not have a cover of size $\sqrt[2k]{\frac{2}{1-\dist(C)}}$,   
	\end{description}
	where $\dist(C)$ is the relative distance of $C$.  
\end{theorem}

Again, it's clear that in the above theorem, starting from a set-system of $n$ sets on $O_k(\log n)$ size universe with no gap and by taking an arbitrary good code with relative distance greater than $1-\nicefrac{1}{(2k)^{2k}}$, we already obtain a $k$-\setcover instance with constant gap, and the universe size  is blown up to merely $(\log n)^{O_k(1)}$. 

In fact our soundness result is stronger than as stated above. We show that the soundness probability in Theorem~\ref{thm:setcoverintro} is actually at most the `collision number' of  $C$ (denoted by $\Cov(C)$) which informally is the smallest number of codewords needed to have collisions each coordinate (see Definition~\ref{def:cov}). We show that the soundness probability in the above theorem is actually $\Cov(C)$ and that  $\sqrt{\frac{2}{1-\dist(C)}}$ is merely a lower bound on $\Cov(C)$.  Thus using codes for which the value of $\Cov(C)$ is optimal we get inapproximability result matching the parameters of Lin \cite{L19} (see Corollary~\ref{cor:setcoverbest}). It is worth noting that the codes achieving optimality of $\Cov(C)$ are simply objects called perfect hash functions but we view them as codes (Proposition~\ref{prop:covtight}).

\noindent\textbf{Comparison to \cite{L19} and \cite{KLM19}.} We emphasize that our contribution in the above result is in the construction of threshold graphs (using arbitrary error correcting codes) and not in the composition of the threshold graph with the input instance. In particular, Lin showed how to construct one particular threshold graph (using universal sets), and we show how to build them in a general way using any code. On the other hand, comparing Theorem~\ref{thm:setcoverintro} to \cite{KLM19}, we note that it is possible to obtain time lower bounds using \DPF technique for $k$-\setcover instance  with constant gap, and the universe size  blown up to merely $(\log n)^{O_k(1)}$, but this cannot be done using arbitrary good codes as in Theorem~\ref{thm:setcoverintro}. In \cite{KLM19} the code used for generating gap is sensitive to the starting hypothesis. Therefore even  to get constant gap under \seth\ using  \cite{KLM19}, we would still need to use the highly non-trivial algebraic geometric codes. 

\subsection{Our Techniques}\label{sec:proof}

Our main technical contribution is the construction of a class of threshold graphs.  A specific threshold graph was constructed in  \cite{L19} using extremal combinatorial objects called universal sets, and in this work, we show how to construct them in general by starting from  just error correcting codes.   

We note that there are several notions of threshold graphs in literature. A common aspect in all constructions is the threshold property: we want a base graph such that a certain subgraph appears many times in the base graph, whereas a slightly bigger (or different) version of this subgraph does not appear (or appears very few times) in the base graph. We now formally define the threshold graph that we use in this work.

\begin{definition}[Thereshold Graph]\label{def:thresholdintro}
A bipartite graph $G=(A\dot\cup B, E)$, with $A=A_1\dot\cup\cdots \dot\cup A_\ell, B=B_1\dot\cup\cdots \dot\cup B_k$ has the threshold property with collision parameter $h>k$ and soundness parameter $\delta$ if
\begin{description}
\item[Completeness:] For every $b_1,\dots b_k\in B_1\times\cdots\times B_k$ and every $i\in[\ell]$ there exist $a\in A_i$ which is a common neighbor of $b_1,\dots b_k$. 
\item[Collision Property:] Let $X\subseteq B$ such that for every $i\in[\ell]$ we have that exists $a\in A_i$ which is a common neighbor of (at least) $k+1$ vertices in $X$. Then $|X|\ge h$.
\item[Soundness:] For every $j\in[t]$ and every distinct $b\neq b'\in B_j$, for all except $(1-\delta)\ell$ of the parts $i\in[\ell]$, we have that $\neighb(b)\cap\neighb(b')\cap A_i = \emptyset$.
\end{description}
\end{definition}

Notice that a threshold graph $G$ should contain many $(k,\ell)$ bicliques, one for each $k$-tuple in $B_1\times\cdots\times B_k$. The same $G$ should not contain \emph{any} $(k+1,\ell')$ biclique (for some $\ell'\ll \ell$) when the right side has at most one vertex from each $A_i$ ($i\in[\ell]$); this is exactly the Threshold property. 

The soundness property is also similar, we require that for every $b_j\in B_j,b_r\in B_r$, there is a joint neighbor $a_i\in A_i$ for \emph{every} $i\in [\ell]$. For $b,b'\in B_j$, $G$ should not contain a  common neighbor in almost all of $A_i$'s.

We show a construction taking any error correcting code and transforming it into a threshold graph. The construction appears in Section~\ref{sec:threshold}.
\begin{theorem}[Threshold Property; Informal statement of Lemma \ref{lem:threshold}] \label{them:threshold-informal}
Let $C:\Sigma^r\rightarrow  \Sigma^\ell$ be a code of distance $\delta$, then for every integer $t\in\mathbb{N}$, there is polynomial time algorithm creating a graph $G=(A\cup B,E)$ of size $O(t\ell\abs{\Sigma}^t\abs{\Sigma}^r)$, which has the Threshold property with collision parameter $\sqrt{\frac{2}{1-\delta}}$ and soundness parameter $\delta$.
\end{theorem}

The parameters in the informal statement are not optimal, and in fact to match the bounds of \cite{L19} we need the parameters in the formal statement. 

Using our construction when the code $C$ is a random error correcting code (i.e. matching each string $w\in\Sigma^r$ to a random string in $\Sigma^\ell$) gives optimal parameters for large enough $\ell$ (see Section~\ref{sec:rand}). On the other hand, taking a completely random graph does not give a good threshold graph matching our requirements, because the soundness property is very unlikely to happen (see Remark~\ref{rem:random}).

\noindent\textbf{Composition Step of Threshold Graph with Input Graph.} We close this subsection by giving some intuition on how the threshold graph given in Definition~\ref{def:thresholdintro} is used in \TGC. First, we rewrite our initial problem (with no gap) as a problem on some bipartite graph $G_0(U\dot\cup W,E_0)$. This reformulaztion is explicit in the definition of the \maxcover problem and for the \setcover problem we consider the bipartite graph formed between universe and collection of input subsets (the edges representing the membership of a universe element in a subset). We then construct a threshold graph $G(A\dot\cup B,E)$ where we have some canonical bijection between $B$ and $W$. Therefore we now have the tripartite graph $H(U\dot\cup W\dot\cup A,E')$, where the edge set between $U$ and $W$ is $E_0$ and the edge set between $W$ (i.e., $B$) and $A$ is $E$. Given $H$, the goal is then to create a bipartite graph $G_1$ between $U$ and $A$, where the edges depend on $H$ in some way. The resulting graph $G_1$ has to be designed to be a gap instance of the starting problem, where the gap is obtained by using the threshold properties of $G$. This step of constructing $G_1$ from $H$ is the most non-trivial part of the \TGC technique. In Theorem~\ref{thm:gapintro}, we indeed provide a novel way to combine the \maxcover problem having the pseudo-projection property and no gap, with the threshold graph of Definition~\ref{def:thresholdintro}, to obtain an instance of \maxcover problem with gap. For Theorem~\ref{thm:setcoverintro}, we simply use the composition provided by Lin~\cite{L19}.

\subsection{Organization of the Paper}
The paper is organized as follows. In Section~\ref{sec:prelim} we define the problems and hypotheses of relevance to this paper, and also recall some basic notions in coding theory. In Section~\ref{sec:threshold}, we show how to construct threshold graphs from arbitrary error correcting codes. In Sections~\ref{sec:gap}~and~\ref{sec:setcover} resp.,   we show how to compose these threshold graphs with $k$-\maxcover and $k$-\setcover instances resp., in order to create a gap. 
Finally in Section~\ref{sec:open}, we discuss an  important open problem stemming from our work.

\section{Preliminaries}\label{sec:prelim}

For a graph $G=(V,E)$ and a vertex $v\in V$, we denote by $\neighb(v)\subset V$ the set of neighbors of $v$.

\subsection{Problems}\label{sec:problems}

\paragraph{\bf $3$-SAT.}
In the $3$-SAT problem, we are given a \cnf formula $\varphi$ over $n$ variables $x_1,\dots x_n$, such that each clause contains at most $3$ literals. Our goal is to decide if there exist an assignment to $x_1,\dots x_n$ which satisfies $\varphi$.

\paragraph{$(k,h)$-\setcover problem.}
For every $k,h\in\mathbb{N},k<h$, in the $(k,h)$-\setcover problem we receive a universe $\mathcal{U}$ and $k$ collection of sets $\mathcal{S}_1,\ldots ,\mathcal{S}_k$ over $\mathcal{U}$. The goal is to distinguish between the two cases:
\begin{itemize}
	\item There exists $(S_1,\dots S_k)\in\mathcal{S}_1\times \cdots \times \mathcal{S}_k$ such that $\bigcup_{i=1}^k S_i = \mathcal{U}$.
	\item For every $\mathcal{S}'\subset \underset{i\in [k]}\bigcup\mathcal{S}_i$, if  $\abs{\mathcal{S}'}< h$, then $\underset{{S\in\mathcal{S}'}}{\bigcup}S \neq \mathcal{U}$.
\end{itemize}
Notice that the $(k,k+1)$-\setcover problem has no gap, and we are interested in creating a gap for the \setcover problem, starting from a no gap instance.

\paragraph{$k$-\textsf{MaxCover} problem.}
\begin{sloppypar}We now recall the  $k$-\maxcover problem introduced by Chalermsook~et~al.~\cite{CCKLMNT20}.  It is in fact a parameterized version of the label cover problem, where each label cover vertex corresponds to a $k$-\maxcover super-node.\end{sloppypar}

The $k$-\maxcover instance $\Gamma$ consists of a bipartite graph $G=(V\dot\cup W, E)$ such that $V$ is partitioned into $V=V_1\dot\cup \cdots \dot\cup V_k$ and $W$ is partitioned into $W=W_1\dot\cup \cdots \dot\cup W_\ell$. We sometimes refer to $V_{i}$'s and $W_j$'s as {\em left super-nodes} and {\em right super-nodes} of $\Gamma$, respectively. 

A solution to $k$-\maxcover is called a {\em labeling},
which is a subset of vertices $v_1\in V_1,\dots v_k\in V_k$.
We say that a labeling $v_1,.\dots v_k$ {\em covers} a right super-node $W_{i}$, if there exists a vertex $w_{i} \in W_{i}$ which is a joint neighbor of all $v_1,\dots v_k$, i.e. $(v_j,w_i)\in E$ for every $j\in[k]$.
We denote by $\maxcover(\Gamma)$ the maximal fraction of right super-nodes that can be simultaneously covered, i.e.
\begin{align*}
\maxcover(\Gamma) = \frac{1}{\ell} \left(\max_{\text{labeling } v_1,\dots v_k} |\{i \in [\ell] \mid W_i \text{ is covered by } v_1,\dots v_k\}|\right). 
\end{align*}

In exact $k$-\maxcover, the input is a maxcover instance $\Gamma$ and the the goal is to decide whether $\maxcover(\Gamma)=1$ or not.

In the $\epsilon$-gap $k$-\maxcover, on input $\Gamma$ the goal is to distinguish between the two cases:
\begin{itemize}
	\item $\maxcover(\Gamma)=1$.   
	\item $\maxcover(\Gamma)<\epsilon$.  
\end{itemize}

%

\paragraph{$k$-\textsf{Clique} problem}
In the $k$-clique problem we receive a graph $H = (V,E)$ with $\abs{V}=n$, and our goal is to decide if $H$ contains a clique of size $k$, i.e. there exists $v_1,\dots v_k\in V$ such that for every $i\neq j\in [k]$, $(v_i,v_j)\in E$.
\subsection{Hypotheses}

\begin{hypothesis}[{W[1]$\neq$ FPT}]
For any computable function $F:\mathbb{N}\rightarrow\mathbb{N}$, there is no $F(k)\poly(n)$-time algorithm which solves the $k$-$\mathsf{Clique}$ problem over $n$ vertices.
\end{hypothesis}

\begin{hypothesis}[Exponential Time Hypothesis (\eth)~\cite{IP01,IPZ01,Tovey84}] \label{hyp:eth}
There exists $\epsilon > 0$ such that no algorithm can solve 3-\SAT on $n$ variables in time $O(2^{\epsilon n})$. Moreover, this holds even when restricted to formulae in which each variable appears in at most three clauses.
\end{hypothesis}

Note that the original version of the hypothesis from~\cite{IP01} does not enforce the requirement that each variable appears in at most three clauses. To arrive at the above formulation, we first apply the Sparsification Lemma of~\cite{IPZ01}, which implies that we can assume without loss of generality that the number of clauses $m$ is $O(n)$. We then apply Tovey's reduction~\cite{Tovey84} which produces a 3-\cnf instance with at most $3m + n = O(n)$ variables and every variable occurs in at most three clauses. This means that the bounded occurrence restriction is also without loss of generality.

\subsection{Error Correcting Codes}\label{sec:ECC}
\begin{definition}[Distance]
Let $\Sigma$ be finite set and $\ell\in\mathbb{N}$, the distance between $x,y\in \Sigma^\ell$ is
\[ \dist(x,y) = \frac{1}{\ell}\abs{\sett{i\in[\ell]}{x_i\neq y_i}}. \]
\end{definition}
\begin{definition}[Error Correcting Code]
Let $\Sigma$ be finite set, for every $\ell\in\mathbb{N}$ a subset $C:\Sigma^r\rightarrow\Sigma^\ell$ is an error correcting code with message length $r$, block length $\ell$ and relative distance $\delta$ if for every $x,y\in \Sigma^r$, $\dist(C(x),C(y))\geq \delta$. We denote then $\dist(C)=\delta$.
\end{definition}
We sometimes abuse notations and treat an error correcting code as its image, i.e. $C\subset \Sigma^\ell$.

For the purpose of this paper, we introduce a new notion on codes called \emph{collision number}.
\begin{definition}[Collision Number]\label{def:cov}
Let $C\subset \Sigma^\ell$, we say that a subset $S\subset C$ is colliding on coordinate $i\in[\ell]$ if there exists $ x,y\in S$ such that $x_i=y_i$. The collision number of a code $C$, $s=\Cov(C)$, is the smallest integer $s$, for which there exists a set $S\subset C,\abs{S}=s$ which collides on \emph{every} coordinate in $[\ell]$. 
\end{definition}

\begin{prop}\label{prop:cov}
For every error correcting code $C\subseteq \Sigma^\ell$ of relative distance $\delta$ we have,
\[\sqrt{\frac{2}{1-\delta}}\leq \Cov(C)\leq \abs{\Sigma}+1.\]
\end{prop}
\begin{proof}
	Fix a code $C\subset \Sigma^\ell$, we prove the bounds. 
	\paragraph{Upper Bound}
	 Let $S\subseteq C$ be any set of cardinality $\abs{\Sigma} + 1$, and let $i\in[\ell]$ be any coordinate. Since there are $\abs{\Sigma}$ possible values for the $i$th coordinate and $\abs{S}\geq\abs{\Sigma}+1$, by the pigeonhole principle there must be two $x,y\in S$ such that $x_i=y_i$.
	
	\paragraph{Lower bound}
	Let $S\subseteq C$ be a set which has a collision on every coordinate $i\in[\ell]$. Let $x,y\in S$ and  let $L_{x,y}\subseteq [\ell]$ be defined as follows. $$L_{x,y}:=\{i\in[\ell]\mid x_i=y_i\}.$$ Since $x$ and $y$ are codewords of $C$ we have $|L_{x,y}|\le (1-\delta)\cdot \ell$. On the other since $S$ is a covering subset we have $$ \sum_{\substack{x,y\in S\\ x\neq y}}|L_{x,y}|\ge \ell.$$ This implies that $\binom{|S|}{2}\cdot |L_{x,y}|\ge \ell$ or in other words, $\binom{|S|}{2}\cdot (1-\delta)\ge 1$. After rearrangement, we have that $|S|\ge \sqrt{\frac{2}{1-\delta}}$.
\end{proof}

In this work we use Reed Solomon code, although in fact we only use the distance of the code and not any other properties.

\begin{theorem}[Reed-Solomon Codes \cite{RS60}]\label{thm:rs}
For every prime power $q$, and every $r \leq  q$, there exists a code  of message length $r$, block length $q$, and relative distance $1-\frac{r}{q}$.
\end{theorem}

\section{Construction of Threshold Graphs}\label{sec:threshold}

We define a threshold graph, essentially as in \cite{L19} (with an additional soundness property) which is used to create gap instances in later sections.

\begin{definition}[Threshold Graphs from Error Correcting Codes]\label{def:threshold}
For every error correcting code $C:\Sigma^r\rightarrow  \Sigma^\ell$ and integer $t \in\N$ we define the threshold graph $G_{C, t} = (A\cup B,E)$ as follows. The vertex sets are  $A=A_1\dot\cup A_2\dot\cup\cdots \dot \cup A_\ell$, $\forall i\in[\ell], |A_i|=\abs{\Sigma}^t$ and   $B=B_1\dot\cup B_2\dot\cup\cdots \dot \cup B_t$, $\forall j\in[t], |B_j|=\abs{\Sigma}^r$. For every $j\in[t]$, we associate $B_j$ with the set of all codewords in $C$, i.e., each vertex in $B_j$ is a unique codeword in the image of $C$ . Similarly, for every $i\in[\ell]$, we associate $A_i$ with the set $\Sigma^t$. We have an edge between $b\in B_j$ and $a \in A_i$ if and only if $C(b)_i=(a)_j$. 
\end{definition}

Various variants of threshold graphs were studied in Theoretically Computer Science as early as in the works of Babai et al.\ \cite{BGKRSW96}, and also  were later used by Lin \cite{L18} in his work on \biclique. 
We emphasize that the above definition of threshold graphs are a novel contribution of this paper and their properties below are as in \cite{L19} albeit that he  constructed one specific threshold graph for a certain range of parameters using very non-trivial objects such as universal sets, whereas we provide a generic way to construct them using relatively basic objects such as error correcting codes.

\begin{lemma}[Threshold Property]\label{lem:threshold}
Let $C:\Sigma^r\rightarrow  \Sigma^\ell$ be a code, let $t\in\mathbb{N}$, then the following holds for the graph $G_{C,t}$ defined above.
\begin{description}
\item[Completeness:] For every $(b_1,\ldots ,b_t)\in B_1\times \cdots \times B_t$, and every $i\in [\ell]$, there is a unique vertex $a\in A_i$ which is a common neighbor of $b_1,\dots ,b_t$.  
\item[Collision Property:] Let $X\subseteq B$ such that for every $i\in[\ell]$ we have that exists $a\in A_i$ which is a common neighbor of (at least) $t+1$ vertices in $X$. Then $|X|\ge \Cov(C)$.
\item[Soundness:] For every $j\in[t]$ and every distinct $b\neq b'\in B_j$, for all except $(1-\dist(C))\ell$ of the parts $i\in[\ell]$, we have that $\neighb(b)\cap\neighb(b')\cap A_i = \emptyset$.
\end{description}
\end{lemma}
The soundness means that two distinct vertices in $B_j$ don't have any joint neighbor in almost all of the partitions $A_i$ (given that $C$ is a code with large distance).
\begin{proof}
Let $G_{C,t}$ be the threshold graph defined above.
\begin{description}
\item[Completeness:] Fix $(b_1,\dots b_t)\in B_1\times\cdots\times B_t$ and $i\in[\ell]$. Let $x_1,\dots x_t$ be the codewords that are associated to $b_1,\dots b_t$ (may have repetitions). Then the vertex $a=((x_1)_i,(x_2)_i\dots (x_t)_i)\in A_i$ is connected to $b_1,\dots b_t$ by definition. Furthermore, it is the only common neighbor. Let $a'\neq a\in A_i$, and let $j\in [t]$ be a location in which $(a)_j\neq (a')_j$, then $a'$ is not connected to $b_j$, since $(a')_j\neq (x_j)_i = (a)_j$.
%

\item[Collision Property:] Let $X\subset B$ be a set such that for every $i\in[\ell]$, there exist $a_i\in A_i$ such that $\abs{\neighb(a_i)\cap X} \geq t+1$, we prove that $\abs{X}\geq\Cov(C)$.

For every $i\in[\ell]$, we know that $\abs{\neighb(a_i)\cap X}\geq t+1$, and $B$ is divided to $t$ parts, so
there must be $j\in[t]$ such that $\neighb(a_i)$ contains at least two vertices in $X\cap B_j$, formally $\abs{\neighb(a_i)\cap X\cap B_j}\geq 2$. Denote these vertices by $b, b'$, and let $x,x'$ be the codewords associated to $b,b'$. Since both vertices are connected to $a_i$, $(x)_i = (x')_i = (a_i)_j$.

Let $S$ be the set of codewords which are the encoding of elements in $X$: 
\[ S = \sett{x\in C}{x \text{ is an encoding of some }b\in X}. \]
From above, for every $i\in [\ell]$ there must be $x,x'\in S$ such that $x_i=x'_i$, so the set $S$ collides on every coordinate, and $\abs{S}\geq\Cov(C_t)$.

Every element $b\in X$ contributed at most a single element to $S$, so $\abs{X}\geq \abs{S}\geq \Cov(C)$.
\item[Soundness:] Fix some $j\in[t]$ and $b\neq b'\in B_j$. Let $x\neq x'$ be the codewords associated to $b,b'$ respectively. Let $i\in[\ell]$ be a partition in $A$ in which $b,b'$ has a joint neighbor $a$. By definition, this means that $(x)_i = (x')_i = (a)_j$. As $x,x'$ are different codewords in $C$, this can happen for at most $(1-\dist(C))\ell$ of the indices in $[\ell]$. \qedhere
\end{description}
\end{proof}

\subsection{Randomized Constructions}\label{sec:rand}

We show that a random code $C:\Sigma^r\rightarrow\Sigma^\ell$ has good distance and collision number when $\ell$ is large enough.
A random error correcting code $C:\Sigma^r\rightarrow \Sigma^\ell $ is a mapping in which for every $x\in \Sigma^r$,  $C(x)$ is chosen uniformly at random in $\Sigma^\ell$. 

\begin{cor}
	For every $t,r,q\in\mathbb{N}$ and $\ell \geq r2^q$, let $C:[q]^r\rightarrow[q]^\ell$ be a random code, then the  graph $G_{C,t}$ has the threshold property with collision property $\Cov(C)\geq\frac{q}{10}$ and soundness $\Delta(C)\geq 1-\frac{2}{q}$.
\end{cor}
The proof follows directly from Lemma~\ref{lem:threshold}, and the properties of the random code from claims \ref{claim:rand-cov} and \ref{claim:rand-dist}.

\begin{claim}\label{claim:rand-cov}
	Fix $\ell \geq r 2^{\abs{\Sigma}}$, and let $C:\Sigma^r\rightarrow \Sigma^\ell $ be a random code, then with high probability, $\Cov(C)\geq\frac{\abs{\Sigma}}{10}$.
\end{claim}
\begin{proof}
	Denote $\abs{\Sigma}=q$. Fix a set $S\subset \Sigma^r,\abs{S}\leq\frac{q}{10}$. We say that $S$ has no collision on coordinate $i\in [\ell]$, if there are no $x,y\in S$ such that $C(x)_i=C(y)_i$. 
	
	Fix an arbitrary coordinate $i$, we lower bound the probability for $S$ to have no collision at $i$. We enumerate over all elements in $S$, $S=x_1,x_2,\dots x_{|S|}$, and for each $j\in\abs{S}$ take the probability that $(x_j)_i\neq(x_{t})_i$ for all $t<j$. Assuming that there was no collision on $t<j$, and that $(x_j)_i$ is distributed uniformly in $[q]$, this equals exactly $\frac{q-j+1}{q}$, hence:
	\[ \Pr[S \text{ has no collisions at } i] = 1\cdot\frac{q-1}{q}\frac{q-2}{q}\cdots\frac{q-\abs{S}+1}{q}\geq \left(\frac{9}{10}\right)^{\frac{q}{10}}. \]
	
	The probability that there is some $i\in[\ell]$ such that $S$ has collisions at $i$ is at most $(1-\left(\frac{9}{10}\right)^{\frac{q}{10}})^\ell$. 
	
	By union bound over all $S$, there are at most $(q^r)^{\frac{q}{10}}$ sets $S$:
	\[ \Pr\left[\Cov(C)>\frac{q}{10}\right]\leq (q^r)^{\frac{q}{10}}\left(1-\left(\frac{9}{10}\right)^{\frac{q}{10}}\right)^\ell \leq e^{rq \log q} e^{-\ell 0.9^q}. \]
	In our case $\ell \geq r2^q$, so $\ell 0.9^q \gg rq\log q $, so it holds with high probability.
\end{proof}

\begin{claim}\label{claim:rand-dist}
	Let $\ell \geq 4r \abs{\Sigma}\ln\abs{\Sigma}$ and $C:\Sigma^r\rightarrow \Sigma^\ell $ be a random error correcting code, then with high probability $\Delta(C)\geq 1-\frac{2}{\abs{\Sigma}}$.
\end{claim}
\begin{proof}
	Denote $\abs{\Sigma}=q$, in this proof we treat $C$ as the image of the code, i.e. $C\subset \Sigma^\ell$.
	For every $x,y\in C$, $\Exp[\Delta(x,y)]=1-\frac{1}{q}$.
	By a Chernoff bound:
	\[ \Pr[\dist(x,y)\leq 1-\frac{2}{q}] \leq e^{-\frac{1}{8q}\ell}. \]
	Preforming union bound over all pairs $x,y\in C$ (there are $q^{2r}$ such pairs):
	\[ \Pr\left[\min_{x\neq y\in C}\{\Delta(x,y)\}<1-\frac{2}{q}\right] \leq q^{2r}e^{-\frac{1}{8q}\ell}, \]
	with $\ell \geq 4r q \ln q$, with high probability the distance is at least $1-\frac{2}{q}$.
\end{proof}

\begin{remark}\label{rem:random}
We remark that if instead of taking a random error correcting code $C$, we would choose $G$ to be a random graph in the Erd\"os-R\'enyi model (i.e. each edge appears with probability $p$) it would not be possible to get the soundness property. This is because for graphs sampled from the Erd\"os-R\'enyi model, there is no distinction between vertices in the same $B_j$ and vertices in different ones. It is unlikely that for each $b_1,\ldots, b_k\in B_1\times\cdots\times B_k$ there is a full bipartite graph with some $a_1,\ldots ,a_\ell$, but for \emph{two} $b,b'\in B_j$ there is no full bipartite graph with some  $a_1,\ldots ,a_\ell$. 
On the other hand, it is possible to get the collision property for a random graph in the Erd\"os-R\'enyi model, albeit with slightly worse parameters than random codes.
 \end{remark}
 
We now show that there are deterministic codes for which we can obtain improvement above Proposition~\ref{prop:cov} on the collision number and are in fact optimal. We start by defining perfect hash families which have received considerable attention in literature (for example see \cite{FK84,FKS84, AABCHNS92,N94,alon1995color}).

  	\begin{definition}[Perfect Hash Family] \label{def:HM}For every $N,\ell,q\in \mathbb{N}$, we say that ${H}:=\{h_i:[N]\to [q]\mid i\in[\ell]\}$ is a $[N,\ell]_q$-Perfect hash family if for every  subset $T$ of $[N]$, where $|T|\le q$, there exists some $i\in[\ell]$ such that:
\begin{align}
 \forall x,y\in T, x\neq y,\ h_i(x)\neq h_i(y).\label{eq:HM}
\end{align}
Moreover, the computation time of $H$ is defined to be the time needed to output the $\ell\times N$ matrix with entries in $[q]$ whose $(i,x)^{\text{th}}$ entry is simply $h_{i}(x)$ (for $h_i\in H$).
\end{definition}
In other words, $H$ is a $[N,\ell]_q$-Perfect hash family if for every $T\subset [N],|T|\leq q$, there exists a hash function $h\in H$ such that $h$ on inputs in $T$ gets $|T|$ distinct values.

\begin{prop}[Collision number of Perfect hash family]\label{prop:covtight}
Let $N,\ell,q\in \mathbb{N}$, and let ${H}$ be a $[N,\ell]_q$-Perfect hash family.
Then $H$ can be interpreted as a code over alphabet $[q]$ of message length $\log_{q} N$, block length   $\ell$ and collision number $q+1$.
\end{prop}
\begin{proof}
Label the hash functions in $H$ using $[\ell]$. We think of $H$ as a code as follows. For every $x\in [N]$ and $i\in[\ell]$, the $x^{\text{th}}$ codeword's $i^{\text{th}}$ coordinate is the image of the $i^{\text{th}}$ hash function in $H$ on the input $x$. 

To see the claim on the collision number of the aforementioned code, suppose for the contrary assume that there exists $T\subseteq [N]$ of cardinality $q$ such that  for every $i\in [\ell]$ there exists $x,y\in T$ such that $x_i=y_i$. This contradicts \eqref{eq:HM}.
\end{proof}

\section{$k$-\maxcover: Gap Creation by Threhosld Graph Composition}\label{sec:gap}
In this section we show a gap creation technique for $k$-\maxcover with a projection property we call \emph{pseudo projection}. This property is an analog of the projection property of label cover.

\begin{definition}[Pseudo Projection]
A $k$-\maxcover instance $\Gamma=(V\dot\cup W,E)$, $V=V_1\dot\cup \cdots \dot\cup V_k$ and $W=W_1\dot\cup \cdots \dot\cup W_t$ has the pseudo projection property if for every $i\in[k],j\in[t]$, one of the two holds:
\begin{itemize}
\item Every $v\in V_i$ has exactly one neighbor $w\in W_j$.
\item There is a full bipartite graph between $V_i$ and $W_j$.
\end{itemize}
\end{definition}

Below is the main result of this section on gap creation in $k$-\maxcover.

\begin{theorem}\label{thm:gap}
Let $\Gamma_0 = (V\dot\cup W,E_0)$ be a $k$-\maxcover instance with the pseudo projection property, with $V=V_1,\dots V_k$, $W = W_1,\dots W_t$. Let $C:\Sigma^r\rightarrow\Sigma^\ell$ be an error correcting code such that $\abs{\Sigma}^r \geq \abs{W_j}$ for every $j\in[t]$.
Then there exists a reduction in time $O(\abs{\Gamma_0}\ell\abs{\Sigma}^t)$ to a $k$-\maxcover instance $\Gamma(V\dot\cup A, E)$ of size $\abs{V}\ell\abs{\Sigma}^t$ with  $V$ divided into $k$ parts, and $A$ into $\ell$ parts. The new instance $\Gamma$ satisfies
\begin{itemize}
	\item If $\maxcover(\Gamma_0)=1$, then $\maxcover(\Gamma)=1$.
	\item If $\maxcover(\Gamma_0)<1$ , then $\maxcover(\Gamma)\le 1-\dist(C)$.  
\end{itemize}
\end{theorem}

\begin{proof}
Let $G_{C,t}$ be the threshold graph from Definition \ref{def:threshold}, with the error correcting code $C$ and integer $t$.  We compose $\Gamma_0$ with $G_{C,t}$ to create our new instance $\Gamma$.

For every $j\in [t]$, we arbitrarily match every vertex in $w_j\in W_j$ to a vertex $b_j\in B_j$ without repetitions. This can be done since $\abs{W_j}\leq \abs{B_j}$. The new instance $\Gamma$ is defined as follows:
\begin{itemize}
\item The vertex sets are $V$ from $\Gamma_0$, and $A$ from $G_{C,t}$.
\item A vertex $v\in V_i$ is connected to $a\in A_j$ if there exists $w_1\in W_1,\dots w_t\in W_t$ such that $v$ is connected to $w_1,\dots w_t$ in $\Gamma_0$, and $a$ is connected to the matching $b_1,\dots b_t$ in $G_{C,t}$.
\end{itemize}

We prove the reduction parameters and correctness.
\paragraph{Runtime and Size} The size of $\Gamma$ is bounded by $\abs{V}\abs{A} = \abs{V}\ell \abs{\Sigma}^t$.
For the runtime, to create the edges in $\Gamma$, for each $v\in V$ and $a\in A$ we go over all their neighbors in $\Gamma_0,G_{C,t}$ and check if every $W_j$ is covered by a joint neighbor. This can be done in a linear time in $\abs{W}$. Therefore, the runtime of the reduction is bounded by $\abs{V}\abs{A}\abs{W}= O(\abs{\Gamma_0}\ell\abs{\Sigma}^t)$.

\paragraph{Completeness} Assume $\maxcover(\Gamma_0)=1$, let $v_1,\dots v_k\in V_1\times\cdots V_k$ be a covering set, and let $w_1,\dots w_t\in W_1\times\cdots \times W_t$ be the vertices covered by $v_1,\dots v_k$, i.e. there is a full bipartite graph between $v_1,\dots v_k$ and $w_1,\dots w_t$.

Let $b_1,\dots b_t$ the matching vertices to $w_1,\dots w_t$. By the definition of the threshold graph, for every $l\in[\ell]$ there exists a vertex $a_l\in A_l$ which is a common neighbor of $b_1,\dots b_t$. From the composition definition, for every $i\in[k],l\in[\ell]$, $a_l$ is a neighbor of $v_i$. Therefore, for every $l\in[\ell]$, $A_l$ is covered by $v_1,\dots v_k$.

\paragraph{Soundness}  Assume $\maxcover(\Gamma_0)<1$. Fix any labeling $v_1\in V_1,\dots v_k\in V_k$ of $\Gamma_0$. Since $\Gamma_0$ is not satisfiable, $v_1,\dots v_k$ do not cover all of $W$. Let $j\in[t]$ be a super-node not covered by $v_1,\dots v_k$. 

Define $S\subset [k]$ to be all indices $i$ such that there is a function from the set $V_i$ to $W_j$. For every $i\in S$, denoted this function by $f_i:V_i\rightarrow W_j$. The instance $\Gamma_0$ has the pseudo projection property, so for every $i'\notin S$ there is a full bipartite graph between $V_{i'}$ and $W_j$. Since $W_j$ is not covered by $v_1,\dots v_k$, there must be $i_1,i_2\in S$ such that $f_{i_1}(v_{i_1})\neq f_{i_2}(v_{i_2})$. Denote $w=f_{i_1}(v_{i_1})$ and $w'=f_{i_2}(v_{i_2})$.
 
Let $b\neq b'\in B_j$ be the vertices matched to $w,w'\in W_j$. By our composition, any neighbor $a\in\neighb( v_{i_1})$ in $\Gamma$ has to be a neighbor of $b$ in $G_{C,t}$ (since $w=b$ is the only neighbor of $v_{i_1}$ in $W_j$). Similarly, every neighbor of $v_{i_2}$ in $\Gamma$ has to be a neighbor of $b'$ in $G_{C,t}$.

By Lemma \ref{lem:threshold}, the threshold graph $G_{C,t}$ is such that for all except $(1-\dist(C))\ell$ of the indices $l\in \ell$, $\neighb(b)\cap\neighb(b')\cap A_l = \emptyset$. From above, for all these $l$'s, $v_{i_1},v_{i_2}$ don't have a common neighbor in $A_l$, and $A_l$ is uncovered by $v_1,\ldots v_k$.
\end{proof}

Using the above theorem we can prove strong inapproximability results for $k$-\maxcover based on \eth\ and \Wone. The proofs of both the theorems essentially follow from the ideas given in \cite{KLM19} and we defer them to Appendix~\ref{sec:maxcover}.

\begin{theorem}\label{thm:eth-hardness}
	Assuming \eth, there is no $n^{o(k)}$ time algorithm that given a $k$-\maxcover instance $\Gamma(G=(V\dot\cup W,E))$, where $V$ is divided into $k$ parts, can decide between the following two cases:
	\begin{description}
		\item[Completeness:] $\emph{\maxcover}(\Gamma)=1$.   
		\item[Soundness:] $\emph{\maxcover}(\Gamma)\le n^{-O(\frac{1}{k^3})}$.  
	\end{description}
\end{theorem}

\begin{theorem}\label{thm:wone}
	Assuming \Wone, for every  computable function $F:\N\to \N$, there is no $F(k)\poly(n)$ time algorithm that given a $k$-\maxcover instance $\Gamma(G=(V\dot\cup W,E))$, where $V$ is divided into $k$ parts, can decide between the following two cases:
	\begin{description}
		\item[Completeness:] $\emph{\maxcover}(\Gamma)=1$.   
		\item[Soundness:] $\emph{\maxcover}(\Gamma)\leq n^{-\frac{1}{10\sqrt{k}}} $.  
	\end{description}
\end{theorem}

\begin{remark}[Comparison to \cite{KLM19}]\label{rem:KLM}
The proof technique in \cite{KLM19} gives us the exact same statement as in Theorem~\ref{thm:gap} and with identical parameters. The difference being that in \cite{KLM19}, the  \DPF is used, whereas we demonstrate that the result can be established using the \TGC technique as well. Additionally, it is easy to see that the 	$k$-\maxcover problem with pseudo-projection property can be reduced without any loss in parameters to a product space problem $\mathsf{PSP}(f)$ over the multi-equality Boolean function $f$ (see \cite{KLM19} for the definition of the two terms). It is also possible to reverse the direction of this reduction. Since \cite{KLM19} show time lower bounds under \eth\ and \Wone\  for $\mathsf{PSP}(f)$, it is easy to see why the two techniques yield the same result. 
\end{remark}

\section{Inapproximability of Parameterized Set Cover}\label{sec:setcover}

In this section, we prove Theorem~\ref{thm:setcoverintro} formally.

\begin{theorem}\label{thm:setcover}
For every integer $n$ and every code $C:\Sigma^r\to\Sigma^\ell$ of relative distance $\delta$, such that $|\Sigma|^r\ge n$, there is an algorithm running in $\poly(n)$ time   that takes as input an instance  $(\mathcal{U}, \mathcal{S}=\mathcal{S}_1\cup\cdots \cup \mathcal{S}_k)$ of $(k,k+1)$-\emph{\setcover} problem (where $|\mathcal{S}|=n$) and outputs an instance  $(\mathcal{U}', \mathcal{S}'=\mathcal{S}'_1\cup\cdots \cup \mathcal{S}'_k)$ of $(k,h)$-\emph{\setcover} problem such that the following holds. 
\begin{description}
\item[Size:] $|\mathcal S|=|\mathcal S'|$ and $|\mathcal U'|= \ell \cdot |\mathcal U|^{|\Sigma|^{k}}$.
\item[Completeness:] If there exists $(S_1,\dots S_k)\in\mathcal{S}_1\times \cdots \times \mathcal{S}_k$ such that $\bigcup_{i=1}^k S_i = \mathcal{U}$ then there exists $(S'_1,\dots S'_k)\in\mathcal{S}'_1\times \cdots \times \mathcal{S}'_k$ such that $\bigcup_{i=1}^k S'_i = \mathcal{U}'$.
\item[Soundness:]  If there is no cover for $\mathcal{U}$ of size $k$ in $\mathcal{S}$ then there is no  cover for $\mathcal{U}'$ of size $h$ in $\mathcal{S}'$ where $ h:=\Cov(C)\ge \sqrt{\frac{2}{1-\delta}}$.    
\end{description}
\end{theorem}

The proof of the above theorem follows immediately from combining the below lemma proved in \cite{L19} with our lower bound on covering number given in Proposition~\ref{prop:cov}.

\begin{lemma}[\cite{L19}]\label{lem:reduction}
There is an algorithm which, given an integer $k$, an instance  $(\mathcal{U}, \mathcal{S}=\mathcal{S}_1\cup\cdots \cup \mathcal{S}_k)$ of $(k,k+1)$-\emph{\setcover} problem (where $|\mathcal{S}|=n$), and a threshold graph $G_{C,k}$  as described in Definition~\ref{def:threshold}, outputs a $(k,h)$-\emph{\setcover} instance $(\mathcal{U}', \mathcal{S}'=\mathcal{S}'_1\cup\cdots \cup \mathcal{S}'_k)$   with $|\mathcal S'|=|\mathcal S|$ and $|\mathcal U'|=\ell \cdot |\mathcal U|^{|\Sigma|^k}$ in $|\mathcal U|^{|\Sigma|^k}\cdot n^{O(1)}$ time such that
\begin{itemize}\item If there exists $(S_1,\dots S_k)\in\mathcal{S}_1\times \cdots \times \mathcal{S}_k$ such that $\bigcup_{i=1}^k S_i = \mathcal{U}$ then there exists $(S'_1,\dots S'_k)\in\mathcal{S}'_1\times \cdots \times \mathcal{S}'_k$ such that $\bigcup_{i=1}^k S'_i = \mathcal{U}'$.
\item  If there is no cover for $\mathcal{U}$ of size $k$ in $\mathcal{S}$ then there is no  cover for $\mathcal{U}'$ of size $\Cov(C)$  in $\mathcal{S}'$ (follows from Lemma~\ref{lem:threshold}).
\end{itemize}
\end{lemma}

Next, we show that for a specific choice of code $C$, we can achieve the following parameters for $k$-\emph{\setcover} problem.

\begin{cor}\label{cor:setcoverbest}
For every integer $n$, there is an algorithm running in $\poly(n)$ time   that takes as input an instance  $(\mathcal{U}, \mathcal{S}=\mathcal{S}_1\cup\cdots \cup \mathcal{S}_k)$ of $(k,k+1)$-\emph{\setcover} problem (where $|\mathcal{S}|=n$) and outputs an instance  $(\mathcal{U}', \mathcal{S}'=\mathcal{S}'_1\cup\cdots \cup \mathcal{S}'_k)$ of $\left(k,\sqrt[k]{\frac{\log |\mathcal U'|}{\log \log |\mathcal U'|}}\right)$-\emph{\setcover} problem such that the following holds. 
\begin{description}
\item[Size:] $|\mathcal S|=|\mathcal S'|$ and $|\mathcal U'|= n$.
\item[Completeness:] If there exists $(S_1,\dots S_k)\in\mathcal{S}_1\times \cdots \times \mathcal{S}_k$ such that $\bigcup_{i=1}^k S_i = \mathcal{U}$ then there exists $(S'_1,\dots S'_k)\in\mathcal{S}'_1\times \cdots \times \mathcal{S}'_k$ such that $\bigcup_{i=1}^k S'_i = \mathcal{U}'$.
\item[Soundness:]  If there is no cover for $\mathcal{U}$ of size $k$ in $\mathcal{S}$ then there is no  cover for $\mathcal{U}'$ of size $\sqrt[k]{\frac{\log |\mathcal U'|}{\log \log |\mathcal U'|}}$ in $\mathcal{S}'$.    
\end{description}
\end{cor}

Notice that the parameters obtained here match the parameters obtained by Lin \cite{L19} by using universal sets. The proof of the above corollary follows by combining Theorem~\ref{thm:setcover} with the theorem below by setting $q=\sqrt[k]{\frac{\log |\mathcal U'|}{\log \log |\mathcal U'|}}$ and then applying Proposition~\ref{prop:covtight}.

\begin{theorem}[Alon et al. \cite{alon1995color}]
For every $N,q\in \mathbb{N}$ there exists a $[N,2^{O(q)}\cdot \log N]_q$-Perfect hash family that can be computed in time $\tilde{O}_q(N)$.
\end{theorem}

\begin{remark}\label{rem:tight}
Also, notice that in Lemma~\ref{lem:reduction}, starting from a universe of size $O_k(\log n)$, in order to obtain good time lower bounds based on various assumptions such as \seth, \eth, and \Wone, we would like that the new universe size is $n^{O(1)}$. This implies that the alphabet of the code used in gap creation can be at most $O\left(\sqrt[k]{\frac{\log n}{\log\log n}}\right)$. Since the collision number of a code is by Proposition~\ref{prop:cov} at most the alphabet size (plus one), we have that it is not possible to obtain better gaps using Lin's scheme of gap creation for \setcover problem. 
\end{remark}

\section{Open Problem}\label{sec:open}

The main open question that stems from our work is if we could prove Theorem~\ref{thm:gap} when the $k$-\maxcover instance does not have the pseudo-projection property but instead is obtained through the standard \seth-hardness  reduction from $k$-\SAT\ to exact $k$-\maxcover. 

\begin{center}
\textit{Assuming \seth, can we show there is no $n^{k-\varepsilon}$ time algorithm (for some $\varepsilon>0$)\\ for gap $k$-\maxcover (of size $n$)
using the Threshold Graph Composition technique?}
\end{center}

A positive answer to the above question in particular for the case $k=2$ (assuming the supposed reduction runs in near linear  time) will have a lot of consequences in the area of fine-grained complexity. Firstly, it will open the window for Lin's \TGC technique to enter the world of inapproximability in subquadratic time, and might provide a lot of new insights, including the potential resolution of many open problems (for example, the subquadratic hardness of the gap closest pair problem \cite{KM19}). Notice that one advantage of \TGC over \DPF is that, in theory, it can handle monochromatic $k$-\maxcover (where instead of picking $(v_1,\ldots ,v_k)\in V_1\times \cdots \times V_k$, we are allowed to pick any $k$ distinct vertices in $V$). We direct the reader to Section~\ref{sec:deg} for an attempt at trying to circumvent the need for pseudo-projection property for gap creation in
Theorem~\ref{thm:gap}, when $k=2$.

\subsubsection*{Acknowledgements}
We would like to thank Lijie Chen for his detailed comments on an earlier version of the paper.

\bibliographystyle{alpha}
\bibliography{refs}

\appendix

\section{Inapproximability of $k$-\maxcover}\label{sec:maxcover}

In this section we show hardness  of $k$-\maxcover with pseudo projection property under \Wone\ and \eth{}, and use Lemma \ref{thm:gap} to show gap $k$-\maxcover hardness under these hypothesis.
\subsection{W[1]-Hardness of Approximation}\label{sec:clique-reduction}

\begin{lemma}\label{lem:clique-reduction}
	For every integer $t\in\mathbb{N}$ and every graph $H=(V',E')$ over $m$ vertices, where $V'=V'_1\dot\cup V'_2\dot\cup\cdots\dot\cup V'_t$, such that for all $i\in[t],$ $V'_i$ is an independent set of size $\frac{m}{t}$, there is an $O(m^3)$- time reduction which outputs a $k$-\maxcover instance $\Gamma=(V\cup W,E)$ of size $O(m^3)$ with $V$ divided into $k=\binom{t}{2}$ parts ans $W$ into $t$ parts. The $k$-\maxcover instance $\Gamma$ satisfies
	\begin{itemize}
		\item If $H$ contains a $t$-clique, then $\maxcover(\Gamma)=1$.
		\item If $H$ does not contain a $t$-clique, then $\maxcover(\Gamma)<1$.  
	\end{itemize}
	Furthermore, $\Gamma$ has the pseudo projection property.
\end{lemma}
\begin{proof}
Given a graph $H=(V',E')$ over $m$ vertices $V'=V'_1\dot\cup V'_2\dot\cup\cdots\dot\cup V'_t$, such that each $V_i'$ is an independent set, we
denote by $E'_{i,j}$ the set of edges between $V'_i$ and $V'_j$.

We now construct the $k$-\maxcover instance, set $k:=\binom{t}{2}$, and let $\Gamma=(V\dot\cup W, E)$ be a bipartite graph with vertex sets $V=V_1\dot\cup V_2\cdots V_k$ and $W=W_1\dot\cup W_2\cdots W_{t}$. We abuse notation a little and refer to a super node $V_l$ for $l\in[k]$ as $V_{i,j}$, for $i,j\in[t]$.

For all distinct $i,j\in[t]$, we define $V_{i,j}$ to be the set of edges $E'_{i,j}$, $V_{i,j}\overset{\Delta}{=}E'_{i,j}$. For every $i\in[t]$, we set $W_i$ to be $V_i'$, $W_{i}\overset{\Delta}{=}V'_{i}$. The edges in $\Gamma$:
\begin{itemize}
\item Between $W_i$ and $V_{i,j}$, we connect $w\in W_i$ and $e=(v_1,v_2)\in V_{i,j}$  if $w=v_1$ or $w=v_2$.
\item Between $W_i$ and $V_{j,l}$, when $i\neq j,l$, we create a full bipartite graph.
\end{itemize} 

We now prove the properties of the reduction.
\paragraph{Reduction Runtime and Size:}
The size of each $V_{i,j}$ is at most $\frac{m^2}{t^2}$  and the size of each $W_j$ is $\frac{m}{t}$. Since checking edge incidence is a trivial task,  the total runtime is at most $|V|\cdot |W|\le m^3$.

\paragraph{Pseudo Projection:} By definition, every edge $e=E'$, $e\in V_{i,j}$ has exactly one neighbor in $V'_i$ and one in $V_j'$, which corresponds in $\Gamma$ to one neighbor on $W_i$ and one on $W_j$. For every other $W_l$, there is a full bipartite graph in $\Gamma$ between $V_{i,j}$ and $W_l$, so $\Gamma$ has the pseudo projection property.

\paragraph{Completeness:} Suppose $H$ contains a $t$-clique. Each $V_j'$ is an independent set, so the $t$ clique has to contain one vertex in each $V_j'$, 
say $\{v'_1,\ldots ,v'_t\}\in V_1'\times\cdots\times V_t'$. We show that there exist a cover in $\Gamma$ with value $1$. For every distinct $i,j\in[t]$, we pick the label $e=(v'_i,v'_j)\in V_{i,j}$ (it exists because $\{v'_1,\ldots ,v'_t\}$ is a clique). We show that this set covers all of $W$: for any $W_j$, the vertex $v'_j\in W_j$ is connected to all of the edges: it is connected to its adjacent edges $(v_i',v_j')\in V_{i,j}$, and because there is a full bipartite graph, it is connected also to $(v_i',v_l')\in V_{i,l}$ when $j\neq i,l$.

\paragraph{Soundness:}Suppose $H$ does not contain a $t$-clique, then consider any labeling $S$ of $G$. If $S$ covers all the super nodes $W_j$, then there exist $\{v_{i,j}\in V_{i,j}\}_{i,j\in[t]}$ such that they have one common neighbor in each $W_j$, say $w_j$. Consider the set of $t$ vertices: $\{w_1,\ldots ,w_t\}$. For every $i,j\in [t]$ we have that $v_{i,j}$ is an edge in $H$ (as $v_{i,j}$ is a neighbor of $w_i$ and $w_j$ in $G$). Therefore, $\{w_1,\ldots ,w_t\}$ is $t$-clique of $H$ leading to a contradiction. 
\end{proof}

\begin{proof}[Proof of Theorem~\ref{thm:wone}]
We prove the theorem by a reduction from \clique to $k$-\maxcover and then to gap $k$-\maxcover, in a similar way to Theorem \ref{thm:eth-hardness}. Assume towards contradiction that there is an algorithm $\mathcal{A}$ running in time $F(k)\poly(n)$ and solves the gap $k$-\maxcover problem described in the theorem. We show an algorithm running in time $F'(t)\poly(m)$ for solving the $t$-clique problem on $m$ vertices.

The input is a graph $H=(V',E')$, where $V'$ is divided into $t$ parts. Denote $m=|V'|$.
\begin{enumerate}
	\item Use the reduction from Lemma~\ref{lem:clique-reduction} on the graph $H$ to get a $k$-\maxcover instance $\Gamma_0=(V_0\cup W_0,E_0)$, such that $V_0$ is divided into $k=\binom{t}{2}$ parts, and $W_0$ into $t$ parts.\label{item:clique-red}
	\item Let $C:\mathbb{F}_q^{t}\rightarrow\mathbb{F}_q^q$ be the Reed Solomon code, for $q$ a large prime power such that $\frac{1}{100}m^{\frac{1}{t}}\leq q\leq m^{\frac{1}{t}}$. Run the reduction from  Theorem \ref{thm:gap} on $\Gamma_0$ and $C$ and receive a $k$-\maxcover instance $\Gamma$.\label{item:clique-gap}
	\item Run algorithm $\mathcal{A}$ on $\Gamma$, answer like $\mathcal{A}$.
\end{enumerate}

We prove the correctness of the algorithm.
	\begin{description}
	\item[Runtime:] The reduction in Item~\ref{item:clique-red} takes $O(m^3)$ time and outputs a $k$-\maxcover instance $\Gamma_0$ of size at most $O(m^3)$, with the pseudo projection property.
	The gap generating algorithm from Theorem \ref{thm:gap} runs in time at most $|\Gamma_0|q q^{t} \leq m^3 m^{1+\frac{1}{t}}$, and outputs $\Gamma$ which is of size at most $|\Gamma_0|  q q^{t} \leq 2m^5$.
	Denote $|\Gamma|=n$, by our assumption, the runtime of $\mathcal{A}$ is $F(k)\poly(n)$ for a computable function $F$. As $n=\poly(m)$ and $k=\binom{t}{2}$, this is at most $F(t^2)\poly(m)$. Thus, the total runtime runtime of the algorithm in total is $F'(t)\poly(m)$ for a computable function $F'$.
	\item[Correctness:] If $H$ contains a t-clique, then by Lemma~\ref{lem:clique-reduction}, $\Gamma_0$ is a $k$-\maxcover instance with the pseudo projection property, such that $\maxcover(\Gamma_0)=1$. By Theorem \ref{thm:gap} $\maxcover(\Gamma)=1$, and algorithm $\mathcal{A}$ answers YES.
	
	If $H$ does not have a t-clique, then by Lemma~\ref{lem:clique-reduction}, $\Gamma_0$ satisfies $\maxcover(\Gamma_0)<1$. Therefore, by Theorem \ref{thm:gap} $\maxcover(\Gamma)<1-\Delta(C)$. The distance of $C$ is $1-\frac{t}{q}\geq 1- \frac{t}{m^{\frac{1}{t}}}$, in terms of  $|\Gamma|=n$, $\maxcover(\Gamma)\leq 1-\Delta(C)\leq n^{-\frac{1}{5t}}$, so the algorithm $\mathcal{A}$ answers no.\qedhere
\end{description}
\end{proof}

\subsection{ETH Hardness of Approximation}\label{sec:eth-reduction}

\begin{lemma}\label{lem:eth-reduction}
	For every 3-\cnf formula $\varphi$ over $n$ variables such that each variable appears in at most $3$ clauses, and every integer $k\in\mathbb{N}$, there is a $n^2 2^{\frac{6n}{k}}$- time reduction which outputs a $k$-\maxcover instance $\Gamma=(V\cup W,E)$ of size $2^{\frac{6n}{k}}\poly(k)$ with $V$ divided into $k$ parts ans $W$ into $O(k^3)$ parts. The $k$-\maxcover instance $\Gamma$ satisfies
	\begin{itemize}
		\item If $\varphi$ is satisfiable, then $\maxcover(\Gamma)=1$.
		\item If $\varphi$ is not satisfiable, then $\maxcover(\Gamma)<1$.  
	\end{itemize}
	Furthermore, $\Gamma$ has the pseudo projection property.
\end{lemma} 
\begin{proof}
Given a 3-\cnf formula $\varphi$ over $n$ variables and $m\leq 3n$ clauses. Let $C_1\dots C_k$ be a partition of the clauses into $k$ approximately equal sets.

Let $G(V\dot\cup W, E)$ be a bipartite graph with vertex sets $V=V_1\dot\cup V_2\cdots V_k$ and $W=W_1\dot\cup W_2\cdots W_{t}$, where $t=\binom{k}{1}+\binom{k}{2}+\binom{k}{3}$. Each $j\in [t]$ is associated to a subset $J\subset[k],1\leq\abs{J}\leq 3$, in the proof we abuse notation a bit and use $J\in[t]$ while treating $J$ as a subset.

Each $V_i$ contains the set of all partial assignments satisfying the clauses in $C_i$. Note that $|V_i|\le 2^{3n/k}$. For each $J\in[t]$, if $J=\{i_1,i_2,i_3\}$, let $S_J$ be the set of variables appears both in $C_{i_1}$ in $C_{i_2}$ and in $C_{i_3}$. For smaller $J$, $S_J$ contains the set of variables appears exactly on all $C_i$ for $i\in J$ and not in other partitions (such that each variable in $[n]$ belongs exactly to a single $S_J$). The vertex set $W_J$ contains all assignments to $S_J$.
Again note that $|W_J|\le 2^{3n/k}$. 

For every $i\in[k],J\in[t]$, the edges are as follows:
\begin{itemize}
	\item If $i\in J$, then connect every $v\in V_i$ with it's consistent assignment $w\in W_J$.
	\item If $i\notin J$, connect every $v\in V_i$ to every $w\in W_J$.
\end{itemize}

We now prove the properties of the reduction.
\paragraph{Reduction Runtime and Size:}
The size of each $V_i$ is at most $2^{\frac{3n}{k}}$ because its an assignment over at most $\frac{3n}{k}$ variables. Similarly the size of each $W_j$ is at most $2^{\frac{3n}{k}}$. The number of partition is clear from construction.

It takes linear time in $n$ to check if every partial assignment to a partition $C_i$ is satisfying. For each vertex $v\in V_i$ it takes at most time $\abs{W}$ to create all its edges. Therefore the total runtime is at most $n\abs{W}\abs{V} = n2^{\frac{3n}{k}}2^{\frac{3n}{k}}kt \leq n^2 2^{\frac{6n}{k}}$.

\paragraph{Partial Projection:} Each vertex $v\in V_i$ is an assignment to all variables in the clauses in $C_i$. For every $J\ni i$, a vertex $w\in W_J$ is an assignment to the variables in $S_J$, which is a subset of the variables in $C_i$. Therefore, there is exactly a single $w\in W_J$ which is consistent with $u$.

For every $J$ such that $i\notin J$, there is a full bipartite graph between $V_i$ to $W_J$, which matches the definition of the pseudo projection property.

\paragraph{Completeness:}Suppose $\varphi$ is a satisfiable formula, and let $x=x_1,\dots x_n$ be an assignment which satisfies $\varphi$. Let $v_1\in V_1,\dots v_k\in V_k$ be the vertices which represents the restriction of $x$ to each part $C_i$. Similarly, let $w_1\in W_1,\dots w_t\in W_t$ be the restriction of $x$ to each $S_J$. We claim that $W$ is fully covered by $v_1,\dots v_k$, with the nodes $w_1,\dots w_t$. Fix an arbitrary $J\in[t],i\in[k]$, if $i\notin J$ then by definition $v_i$ is connected to all of $W_J$, so $w_J$ is covered by $v_i$. For $i\in J$, since both $v_i,w_J$ are a restriction of $x$, they are consistent and are connected by an edge. Therefore $v_1,.\dots v_k$ covers $W$.

\paragraph{Soundness:}Suppose $\varphi$ is not satisfiable, and let $v_1\in V_1,\dots v_k\in V_k$ to be some labeling of $\Gamma$. By definition, each $v_i$ satisfies all clauses in $C_i$. Since $\varphi$ is not satisfiable, $v_1,\dots v_k$ can't be a restriction of a single assignment, and there must be some $r\in[n]$, $i_1,i_2\in [k]$ such that $v_{i_1},v_{i_2}$ assign different values to $x_r$. Assume towards contradiction that $v_1,\dots v_k$ fully covers $W$, and let $w_1\in W_1,\dots w_t\in W_t$ be the vertices which are the joint neighbors. Let $J\in[t]$ be the subset containing all indices from $k$ in which $x_r$ appears, it must be that $i_1,i_2\in J$ (there might be a third index). The vertex $w_J$ assigns some value to $x_r$, and it's not possible that both $v_{i_1},v_{i_2}$ are consistent with it (as they assign different values to $x_r$). Therefore $w_J$ is not a joint neighbor of $v_1,\dots v_k$ and we have a contradiction.
\end{proof}

\begin{proof}[Proof of Theorem~\ref{thm:eth-hardness}]
	We prove the theorem by a reduction from 3-\SAT to $k$-\maxcover, and then to gap $k$-\maxcover. 
	
	Assume towards contradiction that there exists an algorithm $\mathcal{A}$ which solves the gap $k$-\maxcover problem described in the theorem in time $n^{o(k)}$. In particular, $\mathcal{A}$ runs in time less than $n^{\frac{\epsilon}{10} k}$, where $\epsilon$ is the constant from the \eth{} hypothesis (see Hypothesis \ref{hyp:eth}). We show an algorithm for solving 3-$\sat$ on $m$ variables in time less than $2^{\epsilon m}$, refuting \eth{}.
	
	The input is a 3-\cnf formula  $\varphi$ on $m$ variables, such that each variable appears in at most $3$ clauses.
	\begin{enumerate}
	\item Set $k=\frac{1}{\epsilon}$ and run the reduction from Lemma \ref{lem:eth-reduction} on $\varphi$ with parameter $k$, receiving a $k$-\maxcover instance $\Gamma_0=(V_0\cup W_0,E)$ with $V_0$ divided into $k$ parts and $W_0$ into at most $k^3$ parts.
	\item Let $C:\mathbb{F}_q^{k^3}\rightarrow\mathbb{F}_q^q$ be the Reed Solomon code, for $q$ a large prime power such that $\frac{1}{100}2^{\frac{m}{k^4}}\leq q\leq 2^{\frac{m}{k^4}}$. Run the reduction from  Theorem \ref{thm:gap} on $\Gamma_0$ and $C$ and receive a $k$-\maxcover instance $\Gamma$.\label{item:gap}
	\item Run algorithm $\mathcal{A}$ on $\Gamma$, answer like $\mathcal{A}$.
	\end{enumerate}
	
	We show that the described algorithm solves 3-$\sat$ and runs in time $O(2^{\epsilon m})$. 
	\begin{description}
	\item[Runtime:] $\varphi$ is a 3-\cnf formula on $m$ variables such that each variable appears in at most $3$ clauses. From Theorem \ref{lem:eth-reduction}, the size of $\Gamma_0$ is $2^{\frac{6m}{k}}\poly(k)$, and the reduction runtime is at most $2^{\frac{6m}{k}}m^2$. By Theorem \ref{thm:gap}, $\Gamma$ has size at most $|\Gamma_0|  q q^{k^3} = 2^{\frac{6m}{k}}\poly(k) 2^{\frac{2m}{k}}\leq 2^{\frac{10m}{k}}$,  and the runtime of Item \ref{item:gap} is at most  $|\Gamma_0|q q^{k^3} \leq 2^{\frac{10m}{k}}$. Denote $|\Gamma|=n$, the runtime of $\mathcal{A}$ is at most $n^{\frac{\epsilon}{10} k}$, which is less than $2^{\frac{10m}{k}\frac{\epsilon}{10}k}$. So the total run time of the algorithm is less than $ 2^{\epsilon m}$.
	\item[Correctness:] Suppose $\varphi$ is a satisfiable 3-\cnf, then according to Theorem \ref{lem:eth-reduction}, $\Gamma_0$ is a $k$-\maxcover instance with the pseudo projection property, and $\maxcover(\Gamma_0)=1$. From Theorem~\ref{thm:gap}, $\maxcover(\Gamma)=1$ and algorithm $\mathcal{A}$ should answer YES in this case.
	
	In the case $\varphi$ is not satisfiable, then by Theorem~\ref{lem:eth-reduction}, $\Gamma_0$ has the pseudo projection property and $\maxcover(\Gamma_0)<1$. Therefore, by Theorem \ref{thm:gap} $\maxcover(\Gamma)<1-\Delta(C)$. The distance of Reed Solomon code $C:\mathbb{F}_q^{k^3}\rightarrow\mathbb{F}_q^q$ is $1-\frac{r}{q}\geq 1- 2^{-\frac{m}{k^4}}$. Calculating the gap with respect to $|\Gamma|$, $1-\Delta(C)\leq 2^{-\frac{m}{k^4}} = n^{-O(\frac{1}{k^3})}$, as $|\Gamma|=n=2^{\frac{10m}{k}}$. Therefore, $\mathcal{A}$ should answer NO.\qedhere
	\end{description}
\end{proof}

\section{Gap creation in \maxcover for $k=2$} \label{sec:deg}
In the special case of $k=2$, and where the out degree of vertices in $V$ is small, we can show a gap creating reduction for \maxcover without requiring the partial projection property.
\begin{lemma}
	Let $\Gamma_0=(V\cup W,E_0)$ be a \maxcover instance with $V=V_1\dot\cup V_2$, $W = W_1,\dot\cup\cdots\dot\cup W_t$, and for every $v\in V$ and $j\in t$, $\abs{\mathcal{N}(u)\cap W_j}\leq d$. Let $C:\Sigma^r\rightarrow\Sigma^\ell$ be an error correcting code such that $\abs{\Sigma}^r \geq \abs{W_j}$ for every $j\in[t]$.
	Then there exists a reduction in time $O(\abs{\Gamma_0}\ell\abs{\Sigma}^t)$ to a \maxcover instance $\Gamma(V\cup A, E)$ of size $\abs{V}\ell\abs{\Sigma}^t$ with  $V$ divided into $2$ parts, and $A$ into $\ell$ parts. The new instance $\Gamma$ satisfies
	\begin{itemize}
		\item If $\maxcover(\Gamma_0)=1$, then $\maxcover(\Gamma)=1$.
		\item If $\maxcover(\Gamma_0)<1$ , then $\maxcover(\Gamma)\le d^2(1-\dist(C))$.  
	\end{itemize}
\end{lemma}
Notice that in order to use the reduction we must use an error correcting code with large distance, else the soundness guarantee in the above lemma is meaningless.

\begin{proof}
The proof is essentially the same as the proof of Theorem \ref{thm:gap}, with slight modifications.
Let $G_{C,t}$ be the threshold graph from Definition \ref{def:threshold}, with the error correcting code $C$ and integer $t$.  We compose $\Gamma_0$ with $G_{C,t}$ to create our new instance $\Gamma$.
	
For every $j\in [t]$, we arbitrarily match every vertex in $w_j\in W_j$ to a vertex $b_j\in B_j$ without repetitions. This can be done since $\abs{W_j}\leq \abs{B_j}$. The new instance $\Gamma$ is defined as follows:
	\begin{itemize}
		\item The vertex sets of $\Gamma$ are $V$ from $\Gamma_0$, and $A$ from $G_{C,t}$.
		\item A vertex $v\in V_i$ is connected to $a\in A_j$ if there exists $w_1\in W_1,\dots w_t\in W_t$ such that $v$ is connected to $w_1,\dots w_t$ in $\Gamma_0$, and $a$ is connected to the matching $b_1,\dots b_t$ in $G_{C,t}$.
	\end{itemize}
	
	We prove:
	\paragraph{Runtime and Size} The size $\Gamma$ is $\abs{V}\abs{A}$, where $\abs{A} = \ell\abs{A_i} = \ell\abs{\Sigma}^t$.
	
	To create the new instance, for each $v\in V$ and $a\in A$ we go over all their neighbors and check if every $W_j$ is covered by a joint neighbor. This can be done in a linear time in $\abs{W}$. Therefore, the runtime of the reduction is bounded by $\abs{V}\abs{A}\abs{W}= O(\abs{G_0}\ell\abs{\Sigma}^t)$.

	\paragraph{Completeness} Suppose $\maxcover(\Gamma_0)=1$, let $v_1\in V_1,v_2\in V_2$ be a covering set, and let $w_1,\ldots w_t\in W_1\times\cdots \times W_t$ be the joint neighbors covered by $v_1,v_2$.
	
	Let $b_1,\dots b_t$ the matching vertices to $w_1,\dots w_t$. By the definition of the threshold graph, for every $l\in[\ell]$ there exists a vertex $a_l\in A_l$ which is a common neighbor of $b_1,\ldots b_t$ in $G_{C,t}$. From the composition definition, this $a_l$ is a neighbor of $v_1$ and $v_2$ in $\Gamma$, so $A_l$ is covered by $v_1,v_2$. 
	
	\paragraph{Soundness}  $\maxcover(\Gamma_0)<1$.  Let $v_1\in V_1, v_2\in V_2$ to be some labeling of $\Gamma_0$. Since $\Gamma_0$ is not satisfiable, $v_1,v_2$ does not cover all of $W$. Suppose $W_j$ is not covered by $v_1,v_2$, denote $S_1 = \mathcal{N}(v_1)\cap W_j$ and $S_2 = \mathcal{N}(v_2)\cap W_j$. Because $W_j$ is uncovered by $v_1,v_2$, $S_1\cap S_2=\emptyset$, and the degree bounds promises that $\abs{S_1},\abs{S_2}\leq d$. 
	
	Let $S_1',S_2'\subset B_j$ be the vertices in $G_{C,t}$ which are matched to $S_1,S_2$ respectively. The matching is one to one, so $S_1'\cap S_2'=\emptyset$.
	By our composition, any neighbor $a$ of $v_1$ in $\Gamma$ has to be a neighbor of some vertex $b\in S_1'$ in $G_{C,t}$, and the same for $v_{2}$ and $S_2'$. 
	
	By Lemma \ref{lem:threshold}, the threshold graph $G_{C,t}$ is such that for every $j$ and for every $b\neq b'\in B_j$, all except $(1-\dist(C))\ell$ of the indices $l\in \ell$, $\neighb(b)\cap\neighb(b')\cap A_l = \emptyset$. For every $b\in S_1',b'\in S_2'$, there are at most $(1-\dist(C))\ell$ indices such that $\neighb(b)\cap\neighb(b')\cap A_l \neq \emptyset$. There are at most $d^2$ pairs of $b\in S_1',b'\in S_2'$, so for all except $d^2 (1-\Delta(C))\ell$ of the indices $l\in\ell$, there is no $a\in A_l$ which is a common neighbor of $v_1,v_2$.
\end{proof}

%
%
%

\end{document}